\documentclass[11pt]{amsart}
\usepackage{amsmath}
\usepackage{graphicx}
\usepackage{bm}
\usepackage{enumerate}
\theoremstyle{plain}
\newtheorem{theorem}{Theorem}[section]
\newtheorem{lemma}[theorem]{Lemma}
\newtheorem{prop}[theorem]{Proposition}

\theoremstyle{definition}

\numberwithin{equation}{section}

\newcommand{\R}{\mathbb R}
\renewcommand{\H}{\mathbb H}
\newcommand{\C}{\mathbb C}

\newcommand{\tr}{\operatorname{Tr}}

\begin{document}




\title[The area-angular momentum-charge inequality]{The area-angular momentum-charge inequality for black holes with positive cosmological constant}

\author{Edward T. Bryden}
\author{Marcus A. Khuri}
\address{Department of Mathematics\\
Stony Brook University\\
Stony Brook, NY 11794, USA}
\email{khuri@math.sunysb.edu, ebryden@math.sunysb.edu}


\thanks{M. Khuri acknowledges the support of NSF Grant DMS-1308753.}

\begin{abstract}
We establish the conjectured area-angular momentum-charge inequality for stable apparent horizons
in the presence of a positive cosmological constant, and show that it is saturated precisely for extreme Kerr-Newman-de Sitter horizons. As with previous inequalities of this type, the proof is reduced to minimizing an `area functional' related to a harmonic map energy; in this case maps are from the 2-sphere to the complex hyperbolic plane. The proof here is simplified compared to previous results for less embellished inequalities, due to the observation that the functional is convex along geodesic deformations in the target.
\end{abstract}
\maketitle

\section{Introduction}
\label{sec1} \setcounter{equation}{0}
\setcounter{section}{1}

Motivated in part by black hole thermodynamics, in particular the desire for a nonnegative black hole temperature \cite{Dain}, several inequalities relating the area, angular-momentum, and charge of horizons have been established \cite{HennigAnsorgCederbaum,HennigCederbaumAnsorg,AnsorgHennigCederbaum,AnsorgPfister,Clement,
ClementJaramillo,ClementJaramilloReiris,ClementReirisSimon,JaramilloReirisDain,DainJaramilloReiris,
DainReiris,FajmanSimon,Hollands,Jaramillo}. Inequalities elucidating how a cosmological constant $\Lambda$ constrains these quantities, have also been proved \cite{HaywardShiromizuNakao,Simon}. The most recent in this direction is the result of
Clement, Reiris, and Simon \cite{ClementReirisSimon} who have treated the area-angular momentum inequality with $\Lambda>0$ for axisymmetric stable apparent horizons
\begin{equation}\label{1}
|\mathcal{J}|\leq\frac{A}{8\pi}\sqrt{\left(1-\frac{\Lambda A}{4\pi}\right)\left(1-\frac{\Lambda A}{12\pi}\right)},
\end{equation}
and showed that it is saturated precisely for extreme Kerr-de Sitter black holes. The purpose of the present work is to obtain the most general form of this inequality by including charge, as well as to establish the corresponding rigidity result for extreme Kerr-Newman-de Sitter (KNdS)
horizons.

We take an initial data point of view. Recall that an initial data set $(M, g, k, E, B)$ for the Einstein-Maxwell equations consists of a 3-manifold $M$, Riemannian metric $g$, symmetric 2-tensor $k$ representing
extrinsic curvature, and vector fields $E$ and $B$ which constitute the electromagnetic field.
Let $\mu$ and $J$
be the energy and momentum densities of the matter fields, then the constraint equations are given by
\begin{equation}\label{2}
16\pi\mu = R+(\tr_{g}k)^{2}-|k|^{2}-2\Lambda,\quad\quad\quad
8\pi J = \operatorname{div}_{g}(k-(\tr_{g}k)g),
\end{equation}
where $R$ denotes scalar curvature.
When contributions from the electromagnetic field are removed we have
\begin{equation}\label{2}
\mu_{EM} = \mu-\frac{1}{8\pi}(|E|^2+|B|^2),\quad\quad\quad
J_{EM} = J+\frac{1}{4\pi}E\times B,
\end{equation}
where $(E\times B)_{i}=\epsilon_{ijl}E^{j}B^{l}$ is the cross product with
$\epsilon$ the volume form of $g$. The typical energy condition employed for geometric inequalities associated with such initial data is referred to as the charged dominant energy condition
\begin{equation}\label{2.1}
\mu_{EM}\geq|J_{EM}|.
\end{equation}

Consider a closed 2-dimensional surface $S$ embedded in the initial data, with induced metric $\gamma$ and unit normal $n$ pointing inside $M$.
We say that the surface is axially symmetric if the group of isometries of the Riemannian submanifold $(S,\gamma)$ has a subgroup isomorphic to $U(1)$, and \eqref{3} holds. Let $\eta$ denote the Killing field associated with this
symmetry. It will be assumed that the integral curves of $\eta$ are normalized to have an affine length of $2\pi$. Moreover we require that
\begin{equation}\label{3}
\mathfrak{L}_{\eta}\gamma=\mathfrak{L}_{\eta}k(n,\eta)
=\mathfrak{L}_{\eta}E(n)=\mathfrak{L}_{\eta}B(n)=0,
\end{equation}
where $\mathfrak{L}_{\eta}$ is Lie differentiation.
Axisymmetry allows for a canonical expression \cite{AshtekarBeetleLewandowski,BoothFairhurst,DainKhuriWeinsteinYamada} for the angular momentum associated with the surface $S$, namely
\begin{equation}\label{4}
\mathcal{J}=\frac{1}{8\pi}\int_{S}\left(k(n,\eta)+\psi E(n)-\chi B(n)\right)dA_{\gamma}
\end{equation}
where $\chi$ and $\psi$ are potentials for the electric and magnetic field, respectively, to be defined in the next section.
The first term in the integral is the standard expression arising from the Komar angular momentum, and the remaining parts are included so as to achieve conservation of angular momentum in the Einstein-Maxwell context. In particular, if the full initial data set is axisymmetric and there is no charged matter as well as no nonelectromagnetic momentum density in the Killing direction, then the angular momentum \eqref{4} does not vary \cite{DainKhuriWeinsteinYamada} among surfaces which are homologous to one another. Furthermore the electric and magnetic charge of the surface are given by
\begin{equation}\label{5}
Q_{e}=\frac{1}{4\pi}\int_{S}E(n)dA_{\gamma},\quad\quad
Q_{b}=\frac{1}{4\pi}\int_{S}B(n)dA_{\gamma},
\end{equation}
and the square of the total charge is $Q^2=Q_e^2+Q_b^2$.

Recall that the strength of the gravitational field near the surface $S$ may be measured by the null expansions
\begin{equation}\label{6}
\theta_{\pm}:= H_{S} \pm \operatorname{Tr}_{S}k,
\end{equation}
where $H_{S}$ is the mean curvature with respect to the unit outward normal $n$. The null expansions
measure the rate of change of area for a shell of light emitted by the surface in
the outward future direction ($\theta_{+}$), and outward past direction ($\theta_{-}$). Thus the gravitational
field is interpreted as being strong near $S$ if $\theta_{+} < 0$ or $\theta_{-}< 0$, in which case $S$ is referred to
as a future (past) trapped surface. Future (past) apparent horizons arise as boundaries of
future (past) trapped regions and satisfy the equation $\theta_{+} = 0$ ($\theta_{-}=0$). Apparent horizons
may be thought of as quasi-local notions of event horizons, and in fact, assuming cosmic
censorship, they must generically be contained inside black holes \cite{Wald}.

In analogy with minimal surfaces, apparent horizons come with a notion of stability. In order to define this in the current setting, consider normal variations of the (future) apparent horizon $S$ with variational vector field $\partial_{t}=\varphi n$ where $\varphi\in C^{\infty}(S)$. Then a computation \cite{AnderssonMarsSimon} shows that
\begin{equation}\label{7}
\partial_{t}\theta_{+}|_{t=0}=L\varphi:=-\Delta_{\gamma}\varphi+2\langle X,\nabla \varphi\rangle
+(W+\operatorname{div}_{\gamma}X-|X|^2)\varphi,
\end{equation}
where
\begin{equation}\label{8}
W=K-8\pi(\mu+J(n))-\Lambda-\frac{1}{2}|II|^2,\quad\quad\quad X=k(n,\cdot),
\end{equation}
with $K$ the Gauss curvature of $\gamma$ and $II_{ij}=h_{ij}+k_{ij}$ the null second fundamental form associated with $\theta_{+}$; here $h$ is the second fundamental form of $S\subset M$. Although $L$ is not necessarily self adjoint, it has a real principal eigenvalue
$\lambda_1$ and corresponding positive unique (up to scaling) principal eigenfunction $\varphi_1$. The future apparent horizon $S$ is referred to as stable if $\lambda_1\geq 0$. A similar statement holds for past apparent horizons. We remark that, according to \cite[Section 5]{AnderssonMarsSimon}, this notion of stability is consistent with that which is used in \cite{ClementReirisSimon} to establish \eqref{1}. Our main result is as follows.

\begin{theorem}\label{thm1}
Let $(M,g,k,E,B)$ be an initial data set for the Einstein-Maxwell equations with positive cosmological constant $\Lambda>0$, and let $S\subset M$ be an axisymmetric stable apparent horizon on which the charged dominant energy condition \eqref{2.1} holds. Then
\begin{equation}\label{9}
\mathcal{J}^2+\frac{Q^2}{4}\leq\frac{A^2}{64\pi^2}\left[\left(1-\frac{\Lambda A}{4\pi}\right)\left(1-\frac{\Lambda A}{12\pi}\right)-\frac{2\Lambda Q^2}{3}\right],
\end{equation}
and equality is achieved if and only if $(S,\gamma,k(n,\cdot),E,B)$ arises from an extreme
Kerr-Newman-de Sitter horizon.
\end{theorem}

The inequality \eqref{9} may be derived by requiring nonnegativity of the temperature for KNdS black holes (see\cite{CaldarelliCognolaKlemm}), and was conjectured to hold under the above hypotheses in \cite{ClementReirisSimon}. It can be interpreted as yielding, for a black hole of fixed area, an upper bound on the amount of angular momentum and charge that it may contain; this gives a variational characterization of the extreme KNdS configuration as the unique horizon which optimizes this bound. Theorem \ref{thm1} implies previously established inequalities \cite{HaywardShiromizuNakao,Simon} giving variational characterizations of extreme Schwarzschild-dS and extreme Reissner-Nordstr\"{o}m-dS respectively,
\begin{equation}\label{10}
A\leq\frac{4\pi}{\Lambda},\quad\quad\quad \Lambda A^2-4\pi A+16\pi^2 Q^2\leq 0,
\end{equation}
in addition to \eqref{1} associated with extreme Kerr-dS. These results have been used to show how the cosmological constant constrains the amount of angular momentum and charge within a black hole, for instance they naturally imply the universal bounds
\begin{equation}\label{10.1}
\mathcal{J}^2\leq \frac{\sqrt{3}(2-\sqrt{3})}{16\Lambda^2},\quad\quad\quad
Q^2\leq\frac{1}{4\Lambda}.
\end{equation}
Inequality \eqref{9} further improves such bounds. For instance, by maximizing the right-hand side of \eqref{9} over $A$ and performing some algebra we arrive at
\begin{equation}\label{11}
\frac{\mathcal{J}^2}{(3+4\Lambda Q^2)^{3/2}}+\frac{(3+4\Lambda Q^2)^{1/2}}{48\Lambda^2}
\leq\frac{1}{24\Lambda},
\end{equation}
which reduces to the bounds in \eqref{10.1} by setting either $\mathcal{J}=0$ or $Q=0$.

There is a standard approach to proving area inequalities for stable horizons. Namely, from  stability one may derive a lower bound for the area in terms of an `area functional' related to a harmonic map energy, and the desired inequality arises by minimizing this functional and showing that the infimum is achieved precisely for the relevant extreme stationary vacuum configuration. The proof of Theorem \ref{thm1} follows this basic prescription, with the added difficulty that the area functional also depends on the area as a consequence of having a nonzero cosmological constant. This is problematic in that the area functional is no longer simply a regularized version of a harmonic map energy. In \cite{ClementReirisSimon} this issues was resolved through a clever scaling argument, and the same type of strategy works here as well. Our main contribution with regards to the proof of these type of inequalities is to observe that the minimization procedure may be simplified, and also enhanced by providing a gap lower bound. This is achieved by observing that the area functional is convex along geodesic deformations of the functional within the target symmetric space, which in the current context will be the complex hyperbolic plane $\mathbb{H}_{\mathbb{C}}^2$. Thus, one immediately achieves a unique minimizer through elementary means. This type of argument is motivated by the work of Schoen and Zhou \cite{SchoenZhou} on the mass-angular momentum-charge inequalities.

This paper is organized as follows. In the next section we describe the construction of potentials associated with angular momentum and charge. These are then used in Section \ref{sec3} together with the stability property to derive the appropriate area functional. In Section \ref{sec4} we study the rescaled area functional and show that it possesses a unique minimizer for fixed angular momentum and charge, namely the extreme KNdS horizon. Lastly, various formulas and aspects of the Kerr-Newman-de Sitter black holes are described in the appendix, along with a proof of the existence of canonical coordinates used for axisymmetric geometries on a sphere.

\section{Construction of Potentials}
\label{sec2} \setcounter{equation}{0}
\setcounter{section}{2}

In this section we will derive the expression for potentials associated with electric and magnetic charge, as well as angular momentum. Our approach will be to motivate this construction on the horizon $S$, as the restriction of potentials naturally defined in the bulk $M$ which arise from the study of a related geometric inequality, namely the mass-angular momentum-charge inequality \cite{KhuriWeinstein}. For this it will be necessary to place added restrictions on the initial data, which are ultimately not necessary for the existence of potentials on $S$ (or the validity of Theorem \ref{thm1}), but serve the purpose of allowing the following motivational discussion.
Thus, for the time being we will assume that the axisymmetry of $S$ extends to axisymmetry of $M$, that $M$ is simply connected,
and that there is no charged matter or momentum density in the Killing direction:
\begin{equation}\label{12}
\mathfrak{L}_{\eta}g=\mathfrak{L}_{\eta}k
=\mathfrak{L}_{\eta}E=\mathfrak{L}_{\eta}B=0,\quad\quad\quad
\operatorname{div}_{g}E=\operatorname{div}_{g}B=J_{EM}(\eta)=0.
\end{equation}
Under these conditions it is straightforward to obtain a potential for the electric field on the bulk. To see this observe that
\begin{equation}\label{13}
d(\iota_{\eta}\star E)=\mathfrak{L}_{\eta}\star E-\iota_{\eta} d\star E=0,
\end{equation}
where $\iota$ denotes interior product and $\star$ is the Hodge star. It follows from simple connectivity that there is an electric potential satisfying
\begin{equation}\label{14}
d\bar{\chi}=\iota_{\eta}\star E=\eta^{i} E^{l}\varepsilon_{ijl}dx^{j}.
\end{equation}
Since exterior derivatives commute with pullback, we may restrict this equation to $S$ to find
\begin{equation}\label{15}
d\chi=E(n)\iota_{\eta}\varepsilon^{(2)},
\end{equation}
where $\varepsilon^{(2)}$ is the volume form of $\gamma$ and $\chi=i^{*}\bar{\chi}$ with $i:S\hookrightarrow M$ the inclusion map. Now note that equation \eqref{15} has a solution $\chi\in C^{\infty}(S)$ independent of any hypotheses on $M$, since stability implies that the apparent horizon $S$ is topologically a 2-sphere and hence simply connected. The desired electric potential $\chi$ is then defined as a solution to \eqref{15}. Similarly, we define the magnetic potential to be a solution of the equation
\begin{equation}\label{16}
d\psi=B(n)\iota_{\eta}\varepsilon^{(2)}.
\end{equation}
Note that both $\chi$ and $\psi$ are axisymmetric, as it is clear that $\iota_{\eta}d\chi=\iota_{\eta}d\psi=0$.

In order to construct the angular momentum potential, let $p=k-(\tr_{g}k)g$ be the momentum tensor with associated 1-form
\begin{equation}\label{17}
\mathcal{P}=\star(p(\eta)\wedge\eta)=\iota_{\eta}\star p(\eta).
\end{equation}
Then
\begin{equation}\label{18}
d\mathcal{P}=\mathfrak{L}_{\eta}\star p(\eta)-\iota_{\eta}d\star p(\eta)
=-\iota_{\eta}\star\left[\star d\star p(\eta)\right]
=\iota_{\eta}\star 8\pi J(\eta)=\iota_{\eta}\star\left[8\pi J_{EM}(\eta)-2E\times B(\eta)\right].
\end{equation}
Since $E\times B=\star (E\wedge B)$ and
\begin{equation}\label{19}
\iota_{\eta}\star\left[\iota_{\eta}\star(E\wedge B)\right]
=\iota_{\eta}\star d\bar{\psi}(E)
=\iota_{\eta}\left(\star E\wedge d\bar{\psi}\right)
=(\iota_{\eta}\star E)\wedge d\bar{\psi}=\frac{1}{2}d\left(\bar{\chi}\wedge d\bar{\psi}
-d\bar{\chi}\wedge \bar{\psi}\right),
\end{equation}
we have
\begin{equation}\label{20}
d\left(\mathcal{P}-\bar{\chi} d\bar{\psi}+\bar{\psi} d\bar{\chi}\right)=8\pi \iota_{\eta}\star J_{EM}(\eta)=0.
\end{equation}
It follows that there exists a `charged twist potential' such that
\begin{equation}\label{21}
d\bar{\omega}=\mathcal{P}-\bar{\chi} d\bar{\psi}+\bar{\psi} d\bar{\chi}.
\end{equation}
In analogy with the electromagnetic potentials, we may restrict this equation to $S$ and set $\omega=i^{*}\bar{\omega}$ to find
\begin{equation}\label{22}
d\omega=k(n,\eta)\iota_{\eta}\varepsilon^{(2)}-\chi d\psi+\psi d\chi.
\end{equation}
As above, this equation has a solution $\omega\in C^{\infty}(S)$ independent of any hypotheses on $M$, and thus the desired charged twist potential $\omega$ is then defined to be a solution of \eqref{22}. Note that $\iota_{\eta}d\omega=0$ so that this potential is also axisymmetric.

We now record how the potentials just constructed encode the angular momentum \eqref{4} and charge \eqref{5} of $S$. In what follows, as well as in the remaining sections, we will make use of a convenient coordinate system on $S$. By virtue of the fact that $S$ is axisymmetric and topologically a 2-sphere, there exists a global set of polar coordinates $(\theta,\phi)$, with $\theta\in[0,\pi]$ and $\phi\in[0,2\pi)$, such that $\eta=\partial_{\phi}$ and the metric takes the form
\begin{equation}\label{23}
\gamma=e^{2c-\sigma}d\theta^2+e^{\sigma}\sin^{2}\theta d\phi^2,
\end{equation}
where $\sigma\in C^{\infty}(S)$ depends only on $\theta$ and $c$ is a constant related to the area by $A=4\pi e^{c}$. Note that in order to avoid conical singularities at the two axis points $\Gamma=\{\theta=0,\pi\}$, we must have
\begin{equation}\label{23.1}
1=\lim_{\theta\rightarrow 0}\frac{2\pi\cdot\mathrm{Radius}}{\mathrm{Circumference}}
=\lim_{\theta\rightarrow 0}\frac{\int_{0}^{\theta}e^{c-\sigma/2}d\theta}{e^{\sigma/2}\sin\theta}
=e^{c-\sigma(0)}
\end{equation}
and a similar expression at $\theta=\pi$, so that $\sigma(0)=\sigma(\pi)=c$.
Although arguments for the existence of such a coordinate system have been given previously \cite{AshtekarEnglePawlowskiBroeck,DainReiris}, we provide a detailed proof in Appendix \ref{sec6} which is appropriate for the current setting. In these coordinates the electromagnetic potentials are given by
\begin{equation}\label{24}
\chi'=E(n)e^c \sin\theta,\quad\quad\quad \psi'=B(n)e^c \sin\theta,
\end{equation}
where the prime represents $\frac{d}{d\theta}$. Hence
\begin{equation}\label{25}
Q_e=\frac{1}{4\pi}\int_{S}E(n)dA_{\gamma}=\frac{1}{4\pi}\int_{S}\chi' d\theta\wedge d\phi
=\frac{\chi(\pi)-\chi(0)}{2},
\end{equation}
and similarly
\begin{equation}\label{26}
Q_b=\frac{\psi(\pi)-\psi(0)}{2}.
\end{equation}
As for the angular momentum, we find that the charged twist potential is given in coordinates by
\begin{equation}\label{27}
\omega'=k(n,\eta)e^c \sin\theta-\chi\psi'+\psi\chi',
\end{equation}
and therefore
\begin{equation}\label{28}
\mathcal{J}=\frac{1}{8\pi}\int_{S}\left(k(n,\eta)+\psi E(n)-\chi B(n)\right)dA_{\gamma}
=\frac{1}{8\pi}\int_{S}\omega'd\theta\wedge d\phi
=\frac{\omega(\pi)-\omega(0)}{4}.
\end{equation}

\section{The Area Functional}
\label{sec3} \setcounter{equation}{0}
\setcounter{section}{3}

In this section the horizon stability condition will be used to derive a lower bound for the area in terms of a set of quantities related to harmonic maps from $\mathbb{S}^2\rightarrow\mathbb{H}_{\mathbb{C}}^2$. From now on it will be assumed that the surface $S$ is a stable future apparent horizon; similar arguments hold if $S$ is a stable past apparent horizon. Stability asserts that the principal eigenvalue of the stability operator \eqref{7} is nonnegative. Therefore, if $\varphi_1$ denotes the positive principal eigenfunction then for any test function $v\in C^{\infty}(S)$ we have
\begin{align}\label{29}
\begin{split}
0\leq &\int_{S}v^{2}\varphi_{1}^{-1}L\varphi_{1}\\
=&\int_{S}\nabla\varphi_{1}\cdot\nabla(\varphi_{1}^{-1}v^2)
+2\left(X\cdot\nabla\varphi_{1}\right)\varphi_{1}^{-1}v^2
+\left(W+\operatorname{div}_{\gamma}X-|X|^2\right)v^2\\
=&\int_{S}-\left(|\nabla\log\varphi_{1}|^{2}-2X\cdot\nabla\log\varphi_1
+|X|^2 \right)v^2+2v\left(\nabla\log\varphi_{1}-X\right)\cdot\nabla v+Wv^2.
\end{split}
\end{align}
Let
\begin{equation}\label{30}
e_1=n,\quad\quad e_2=e^{\sigma/2-c}\partial_{\theta},\quad\quad
e_3=\frac{1}{e^{\sigma/2}\sin\theta}\partial_{\phi},
\end{equation}
be an orthonormal frame on $S$, and consider an axisymmetric test function so that $e_{3}(v)=0$.
This, together with $\mathfrak{L}_{\eta}k(n,\eta)=0$ from \eqref{3}, and an integration by parts shows that
\begin{equation}\label{31}
\int_{S}X(e_{3})e_{3}(\log\varphi_{1})v^2=0.
\end{equation}
Then \eqref{29} implies that
\begin{equation}\label{32}
0\leq\int_{S}|\nabla v|^2+Wv^2 -X(e_3)^2 v^2-|e_{3}(\log\varphi_{1})|^2v^2-\left|(e_{2}(\log\varphi_1)-X(e_2))v-e_{2}(v)\right|^2.
\end{equation}

\begin{lemma}\label{stability}
Under the hypotheses of Theorem \ref{thm1}, for any axisymmetric $v\in C^{\infty}(S)$ the following stability inequality holds
\begin{equation}\label{33}
\int_{S}\left(|\nabla v|^2+Kv^2\right)dA_{\gamma}
\geq\int_{S}\left(k(n,e_{3})^2+E(n)^2+B(n)^2+\Lambda\right)v^2dA_{\gamma}.
\end{equation}
\end{lemma}

\begin{proof}
In light of \eqref{32}, the desired inequality follows from the charged dominant energy condition \eqref{2.1} and the computation
\begin{equation}\label{34}
\mu+J(n)=\mu_{EM}+J_{EM}(n)+\frac{1}{8\pi}(|E|^2+|B|^2)-\frac{1}{4\pi}E\times B(n)
\geq\frac{1}{8\pi}\left(E(n)^2+B(n)^2\right).
\end{equation}
\end{proof}

In order to choose an appropriate test function $v$, we rely on the intuition that the stability inequality \eqref{33} should be saturated for the extreme KNdS black hole. By analyzing the second variation of area one may verify that this is indeed the case for a test function of the form below. We then choose
\begin{equation}\label{35}
v=\sqrt{\zeta_{a}}e^{-\sigma/2},\quad\quad\quad\zeta_{a}=1+\frac{a^{2}\Lambda}{3}\cos^{2}\theta,
\end{equation}
where $a$ is a constant to be specified. Each term of \eqref{33} will be computed separately. Observe that \eqref{24} and \eqref{27} yield
\begin{equation}\label{36}
\int_{S}\left(E(n)^2+B(n)^2\right)v^2 dA_{\gamma}
=\int_{\mathbb{S}^2}e^{-c}\zeta_{a}\frac{e^{-\sigma}}{\sin^2\theta}\left(\chi'^2+\psi'^2\right)dA,
\end{equation}
and
\begin{equation}\label{37}
\int_{S}k(n,e_{3})^2v^2 dA_{\gamma}
=\int_{\mathbb{S}^2}e^{-c}\zeta_{a}\frac{e^{-2\sigma}}{\sin^4\theta}\left(\omega'
+\chi\psi'-\psi\chi'\right)^2 dA,
\end{equation}
where $dA$ is the area form on the round sphere $\mathbb{S}^2$. Furthermore, calculations show that the Gauss curvature is given by
\begin{equation}\label{38}
K=\frac{e^{\sigma-2c}}{\sin\theta}\left[\sin\theta-\sigma'\cos\theta-\frac{1}{2}\sigma'^{2}\sin\theta
-\frac{1}{2}(\sin\theta\sigma')'\right],
\end{equation}
and
\begin{equation}\label{39}
|\nabla v|^2=\frac{e^{-2c}}{4}\left(\frac{\zeta_{a}'^{2}}{\zeta_{a}}-2\zeta_{a}\sigma'+\zeta_{a}
\sigma'^2\right).
\end{equation}
It follows that
\begin{equation}\label{40}
\int_{S}(|\nabla v|^2+Kv^2)dA_{\gamma}
=\int_{\mathbb{S}^2}e^{-c}\left[\frac{\zeta_{a}'^2}{4\zeta_{a}}+\zeta_{a}
-\frac{\zeta_{a}'\sigma'}{2}-\frac{\zeta_{a}\sigma'^2}{4}
-\frac{\zeta_{a}\sigma'\cos\theta}{\sin\theta}-\frac{\zeta_{a}(\sin\theta\sigma')'}{2\sin\theta}\right]
dA.
\end{equation}
Integrating the last two terms by parts produces
\begin{align}\label{41}
\begin{split}
&\int_{\mathbb{S}^2}e^{-c}\left[-\frac{\zeta_{a}\sigma'\cos\theta}{\sin\theta}
-\frac{\zeta_{a}(\sin\theta\sigma')'}{2\sin\theta}\right]dA\\
=&\int_{\mathbb{S}^2}e^{-c}\left[\frac{\zeta_{a}'\sigma'}{2}
+(\zeta_{a}'\cot\theta-\zeta_{a})\sigma\right]dA
+2\pi e^{-c}\left(\zeta_{a}\sigma(0)+\zeta_{a}\sigma(\pi)\right),
\end{split}
\end{align}
so that
\begin{equation}\label{42}
\int_{S}(|\nabla v|^2+Kv^2)dA_{\gamma}
=4\pi c\alpha_{a}e^{-c}
+\int_{\mathbb{S}^2}e^{-c}\left(\frac{\zeta_{a}'^2}{4\zeta_{a}}+\zeta_{a}
-\frac{\zeta_{a}\sigma'^2}{4}-(1+\Lambda a^2\cos^{2}\theta)\sigma\right)dA,
\end{equation}
where
\begin{equation}\label{43}
\alpha_{a}=\zeta_{a}(0)=\zeta_{a}(\pi)=1+\frac{\Lambda a^2}{3}
\end{equation}
and we have used
\begin{equation}\label{44}
\zeta_{a}'\cot\theta-\zeta_{a}=-(1+\Lambda a^2\cos^2\theta).
\end{equation}
By combining \eqref{36}, \eqref{37}, and \eqref{42}, the stability inequality Lemma \ref{stability} yields
\begin{equation}\label{45}
4\pi c \alpha_{a}+\beta_{a}\geq \mathcal{I}_{a}(\Psi),
\end{equation}
where
\begin{equation}\label{46}
\beta_{a}=\int_{\mathbb{S}^2}\left(\frac{\zeta_{a}'^2}{4\zeta_{a}}+\zeta_{a}\right)dA,
\quad\quad\quad\Psi=(\sigma,\omega,\chi,\psi),
\end{equation}
and
\begin{align}\label{47}
\begin{split}
\mathcal{I}_{a}(\Psi)=&\int_{\mathbb{S}^2}(1+\Lambda a^2\cos^2\theta)\sigma dA\\
&
+\int_{\mathbb{S}^2}\zeta_{a}\left(\frac{\sigma'^2}{4}+\frac{e^{-2\sigma}}{\sin^4\theta}(\omega'+\chi\psi'-\psi\chi')^2
+\frac{e^{-\sigma}}{\sin^2\theta}(\chi'^2+\psi'^2)
+\Lambda\left(\frac{A}{4\pi}\right)^2 e^{-\sigma}\right)dA.
\end{split}
\end{align}
Finally, by recalling that $A=4\pi e^c$, \eqref{45} may be rewritten as
\begin{equation}\label{48}
A\geq 4\pi e^{\frac{\mathcal{I}_{a}(\Psi)-\beta_{a}}{4\pi \alpha_{a}}}.
\end{equation}

Inequality \eqref{48} is the desired area lower bound which will play a central role in the proof of Theorem \ref{thm1}. Typically when establishing geometric inequalities in the spirit of \eqref{9}, after a lower bound has been achieved for the area in terms of a functional related to a harmonic map energy, the next step is to show that the functional attains a global minimum at an appropriate extreme black hole configuration. This is fairly straightforward when the cosmological constant is not present, since in that case the functional $\mathcal{I}_{a}$ is simply a renormalized harmonic map energy. In the current situation, when $\Lambda\neq 0$, the primary difficulty arises from the fact that $\mathcal{I}_{a}$ depends on the area $A$. A consequence of this is that an infimum may not exist when minimizing the functional over all maps $\Psi$ with fixed angular momentum $\mathcal{J}$ and charge $Q$, since the triple $(A,\mathcal{J},Q)$ may not arise from an extreme KNdS black hole. A similar situation occurs in \cite{ClementReirisSimon}, and is resolved with a scaling argument which we now generalize.

\begin{lemma}\label{scaling}
Given $(A,\mathcal{J},Q)\in\mathbb{R}^3_{+}$, there exists a unique $(\hat{A},\hat{\mathcal{J}},\hat{Q})\in\mathbb{R}^{3}_{+}$ which saturates \eqref{9} and satisfies
\begin{equation}\label{49}
\hat{\mathcal{J}}=\frac{\mathcal{J}}{A^2}\hat{A}^2,\quad\quad\quad
\hat{Q}=\frac{Q}{A}\hat{A},\quad\quad\quad \hat{A}\leq\frac{4\pi}{\Lambda}.
\end{equation}
Moreover, inequality \eqref{9} is equivalent to the inequality $\hat{A}\geq A$.
\end{lemma}

\begin{proof}
Consider the curve in $\mathbb{R}^3_{+}$ given by
\begin{equation}\label{50}
f(\tau)=(A(\tau),\mathcal{J}(\tau),Q(\tau))
=\left(\tau,\frac{\mathcal{J}}{A^2}\tau^2,\frac{Q}{A}\tau\right).
\end{equation}
For small $\tau$ the two sides of \eqref{9} have the asymptotics
\begin{equation}\label{51}
\mathcal{J}^2(\tau)+\frac{Q^4(\tau)}{4}\sim \tau^4,\quad\quad
\frac{A^2(\tau)}{64\pi^2}\left[\left(1-\frac{\Lambda A(\tau)}{4\pi}\right)\left(1-\frac{\Lambda A(\tau)}{12\pi}\right)-\frac{2\Lambda Q^2(\tau)}{3}\right]\sim \tau^2.
\end{equation}
Thus, for small $\tau$ inequality \eqref{9} holds when restricted to the curve $f$, although for large $\tau$ it is clear that the opposite inequality holds. It follows that there exists a time $\tau=\hat{A}$ for which \eqref{9} is saturated. Further analysis of the zeros of the associated quartic equation show that this time is unique among those for which $\tau=A(\tau)\leq\frac{4\pi}{\Lambda}$.

Lastly, the inequality $\hat{A}\geq A$ holds if and only if the point $(A,\mathcal{J},Q)$ lies below the surface in $\mathbb{R}^3_{+}$ defined by equality in \eqref{9}; here `below' refers to the interpretation of the $\mathcal{J}$-axis as measuring height. Therefore $\hat{A}\geq A$ if and only if the inequality \eqref{9} holds.
\end{proof}

The fact that $(\hat{A},\hat{\mathcal{J}},\hat{Q})$ saturates \eqref{9} implies that these values for the area, angular momentum, and charge arise from an extreme KNdS solution. This particular extreme KNdS solution yields a map $\Psi_0$ (see Appendix \ref{sec5}) which is a candidate minimizer for a rescaled version of the functional in \eqref{47}. To construct the rescaled functional let $\mathcal{M}(\hat{\mathcal{J}},\hat{Q})$ denote the mass of the extreme KNdS black hole, and set
\begin{equation}\label{52}
\hat{a}=\frac{\hat{\mathcal{J}}}{\mathcal{M}(\hat{\mathcal{J}},\hat{Q})},\quad\quad\quad
\hat{\Psi}=(\hat{\sigma},\hat{\omega},\hat{\chi},\hat{\psi})
=\left(\sigma+\log\frac{\hat{A}}{A},\frac{\hat{A}^2}{A^2}\omega,\frac{\hat{A}}{A}\chi,
\frac{\hat{A}}{A}\psi\right).
\end{equation}
Note that $\hat{a}$ is the value of the parameter $a$ in the extreme KNdS solution (Appendix \ref{sec5}) having angular momentum $\hat{\mathcal{J}}$ and charge $\hat{Q}$, and moreover
\begin{equation}\label{52.1}
\hat{\mathcal{J}}=\frac{\hat{\omega}(\pi)-\hat{\omega}(0)}{4},\quad\quad\quad
\hat{Q}_e=\frac{\hat{\chi}(\pi)-\hat{\chi}(0)}{2},\quad\quad\quad
\hat{Q}_b=\frac{\hat{\psi}(\pi)-\hat{\psi}(0)}{2}.
\end{equation}
A calculation shows that
\begin{equation}\label{53}
\mathcal{I}_{\hat{a}}(\Psi)=\mathcal{I}_{\hat{a}}(\hat{\Psi})
-2\log\frac{\hat{A}}{A}\int_{\mathbb{S}^2}(1+\Lambda\hat{a}^2\cos^2\theta)dA
=\mathcal{I}_{\hat{a}}(\hat{\Psi})-8\pi \alpha_{\hat{a}}\log\frac{\hat{A}}{A},
\end{equation}
and therefore \eqref{48} becomes
\begin{equation}\label{54}
A\geq 4\pi e^{\frac{\mathcal{I}_{\hat{a}}(\hat{\Psi})-\beta_{\hat{a}}}{4\pi \alpha_{\hat{a}}}}
\frac{A^2}{\hat{A}^2}.
\end{equation}

\begin{lemma}\label{lemma3}
The area-angular momentum-charge-$\Lambda$ inequality \eqref{9} holds if
$\mathcal{I}_{\hat{a}}(\hat{\Psi})\geq\mathcal{I}_{\hat{a}}(\Psi_0)$.
\end{lemma}

\begin{proof}
Computation of the area of the extreme KNdS horizon yields
\begin{equation}\label{56}
4\pi e^{\frac{\mathcal{I}_{\hat{a}}(\Psi_0)-\beta_{\hat{a}}}{4\pi \alpha_{\hat{a}}}}=\hat{A}.
\end{equation}
From \eqref{54} we then have $\hat{A}\geq A$, and the desired result follows Lemma \ref{scaling}.
\end{proof}

\section{Minimization and the Proof of Theorem \ref{thm1}}
\label{sec4} \setcounter{equation}{0}
\setcounter{section}{4}

In this section we study the minimization properties of the functional $\mathcal{I}_{a}$, when the parameters $a$ and $A$ defining the functional arise from an extreme KNdS black hole; it will be extremized over the space of maps $\Psi=(\sigma,\omega,\chi,\psi):\mathbb{S}^2\rightarrow\mathbb{H}^{2}_{\mathbb{C}}$ having angular momentum and charge $\mathcal{J}$, $Q_{e}$, and $Q_{b}$ arising from the same extreme KNdS solution. The fundamental reason behind the success of the minimization procedure to follow, is the fact that $\mathcal{I}_{a}$ is closely related to a harmonic map energy. To give the precise relationship, let $\Omega\subset\mathbb{S}^2$ be a domain which does not intersect the axis $\Gamma$, and consider the functional $\mathcal{I}_{\Omega}(\Psi)$ which is obtained from \eqref{47} by restricting the domain of integration to $\Omega$. Let $u=-\sigma/2-\log\sin\theta$ and set
$\tilde{\Psi}=(u,\omega,\chi,\psi):\mathbb{S}^{2}\setminus\Gamma\rightarrow
\mathbb{H}_{\mathbb{C}}^{2}$, then the quasi-harmonic map energy over $\Omega$ is given by
\begin{equation}\label{57}
E_{\Omega}(\tilde{\Psi})=\int_{\Omega}\zeta_{a}\left(u'^2+e^{4u}(\omega'+\chi\psi'-\psi\chi')^2
+e^{2u}(\chi'^2+\psi'^2)+\Lambda\left(\frac{A}{4\pi}\right)^2e^{2u}\sin^2\theta\right)dA.
\end{equation}
Recall that the complex hyperbolic plane $\H^2_\C$ is the homogeneous Riemannian manifold $(\R^4,h_0)$ with metric
\begin{equation}\label{58}
h_0 = du^2 + e^{4u} (dv + \chi d\psi - \psi d\chi)^2 + e^{2u} (d\chi^2 + d\psi^2),
\end{equation}
and therefore the pseudo-energy \eqref{57} differs from the harmonic energy by the factor $\zeta_{a}$ and the last term involving $\Lambda$. Now integrate by parts and use \eqref{44} in the form
\begin{equation}\label{59}
\operatorname{div}_{\mathbb{S}^2}\left(\zeta_{a}\nabla\log\sin\theta\right)
=-(1+\Lambda a^2\cos^2\theta)
\end{equation}
to obtain
\begin{equation}\label{60}
\mathcal{I}_{\Omega}(\Psi)
=E_{\Omega}(\tilde{\Psi})-\int_{\Omega}(1+\Lambda a^2\cos^2\theta)\log\sin\theta dA
-\int_{\partial\Omega}\zeta_{a}(\sigma+2\log\sin\theta)\partial_{\nu}\log\sin\theta ds,
\end{equation}
where $\nu$ is the unit outer normal to $\partial\Omega$. This shows that $\mathcal{I}_{a}$ may
be considered as a regularization of $E$ since the infinite term $\int\zeta_a(\log\sin\theta)'^2$ has been removed. Furthermore, since the two functionals differ only by a boundary term and a constant, they must have the same critical points.

Let $\Psi_{0}=(\sigma_{0},\omega_{0},\chi_{0},\psi_{0})$ be the renormalized map arising from
the extreme KNdS solution which is associated with the functional $\mathcal{I}_{a}$. As is shown in Appendix \ref{sec5}, $\Psi_0$ is a critical point of $\mathcal{I}_a$.  It is the purpose of this section to show that $\Psi_{0}$ realizes the global minimum for $\mathcal{I}_a$.

\begin{theorem}\label{thm2}
Suppose that $\Psi=(\sigma,\omega,\chi,\psi)$ is smooth and satisfies the asymptotics \eqref{66}
with $\omega|_{\Gamma}=\omega_{0}|_{\Gamma}$, $\chi|_{\Gamma}=\chi_{0}|_{\Gamma}$, $\psi|_{\Gamma}=\psi_{0}|_{\Gamma}$, then for any $p\geq 1$ there
exists a constant $C>0$ such that
\begin{equation}\label{61}
\mathcal{I}_a(\Psi)-\mathcal{I}_a(\Psi_{0})
\geq C\int_{\mathbb{S}^{2}}
\left(\operatorname{dist}_{\mathbb{H}_{\mathbb{C}}^{2}}(\tilde{\Psi},\tilde{\Psi}_{0})
-D\right)^2dA,
\end{equation}
where $D$ is the average value of $\operatorname{dist}_{\mathbb{H}_{\mathbb{C}}^{2}}(\tilde{\Psi},\tilde{\Psi}_{0})$.
\end{theorem}

The proof of this result is based on convexity of the quasi-harmonic energy under geodesic deformations; such a property is well-known for the pure harmonic energy
when the target space is nonpositively curved. To explain how this works,
let $\Omega_{\varepsilon}=\{(\theta,\phi)\in\mathbb{S}^2\mid \sin\theta>\varepsilon\}$.
Then with a cut-and-paste argument it will be shown that we may assume that
$\Psi$ satisfies
\begin{equation}\label{62}
\operatorname{supp}(\omega-\omega_{0},\chi-\chi_{0},\psi-\psi_{0})\subset \Omega_{\varepsilon}.
\end{equation}
Next, let $\tilde{\Psi}_{t}$, $t\in[0,1]$ be a geodesic in $\mathbb{H}_{\mathbb{C}}^{2}$ connecting
$\tilde{\Psi}_{1}=\tilde{\Psi}$ and $\tilde{\Psi}_{0}$, this means that for each $(\theta,\phi)$ in the domain, $t\rightarrow\tilde{\Psi}_{t}(\theta,\phi)$ is a geodesic. It then follows that
$(\omega_{t},\chi_{t},\psi_{t})\equiv(\omega_{0},\chi_{0},\psi_{0})$
on $\mathbb{S}^2\setminus\Omega_{\varepsilon}$, so that in particular $\sigma_{t}=\sigma_{0}+t(\sigma-\sigma_{0})$ on this
domain. The fact that $\sigma_{t}$ is linear together with convexity of the quasi-harmonic energy yields
\begin{equation}\label{63}
\frac{d^{2}}{dt^{2}}\mathcal{I}_{a}(\Psi_{t})
\geq 2\int_{\mathbb{S}^2}
|\nabla\operatorname{dist}_{\mathbb{H}^{2}_{\mathbb{C}}}(\tilde{\Psi},\tilde{\Psi}_{0})|^{2}dA.
\end{equation}
Furthermore, since $\Psi_{0}$ is a critical point
\begin{equation}\label{64}
\frac{d}{dt}\mathcal{I}_{a}(\Psi_{t})|_{t=0}=0.
\end{equation}
Theorem \ref{thm2} may then be established by integrating \eqref{63} and applying the Poincar\'{e} inequality. In the remainder of this section we will justify each of these steps.

Before proceeding we record the appropriate asymptotic behavior of $\Psi$. Our assumptions here are based on the asymptotics ($\theta\rightarrow 0,\pi$) of the extreme KNdS map $\Psi_0$, which are given by
\begin{equation}\label{65}
\sigma_0|_{\Gamma}=\log\left(\frac{A}{4\pi}\right),\text{ } \omega_0,\chi_0,\psi_0=O(1),\text{ }
\sigma_0', \omega_0',\chi_0',\psi_0'=O(\sin\theta),\text{ }
\omega_0'+\chi_0\psi_0'-\psi_0\chi_0'=O(\sin^{3}\theta).
\end{equation}
We then require that $\Psi$ satisfies
\begin{equation}\label{66}
\sigma|_{\Gamma}=\log\left(\frac{A}{4\pi}\right),\quad \omega,\chi,\psi=O(1),\quad
\sigma', \omega',\chi',\psi'=O(\sin\theta),\quad
\omega'+\chi\psi'-\psi\chi'=O(\sin^{1+\delta}\theta),
\end{equation}
for some $\delta>0$.

In order to carry out the proof of Theorem \ref{thm2} as outlined above, we first show that it is possible to
approximate $\mathcal{I}_a(\Psi)$ by replacing $\Psi$ with a map $\Psi_{\varepsilon}$ which satisfies \eqref{62}. This is achieved with a cut and paste argument. Define a Lipschitz cut-off function
\begin{equation}\label{67}
\varphi_{\varepsilon}(\theta)=\begin{cases}
0 & \text{ if $\sin\theta\leq\varepsilon$,} \\
\frac{\log\left(\frac{\sin\theta}{\varepsilon}\right)}{\log\left(\frac{\sqrt{\varepsilon}}{
\varepsilon}\right)} &
\text{ if $\varepsilon<\sin\theta<\sqrt{\varepsilon}$,} \\
1 & \text{ if $\sin\theta\geq\sqrt{\varepsilon}$,} \\
\end{cases}
\end{equation}
and let
\begin{equation}\label{68.0}
\Psi_{\varepsilon}=(\sigma,\omega_{\varepsilon},\chi_{\varepsilon}
,\psi_{\varepsilon}),\quad\quad\quad
(\omega_{\varepsilon},\chi_{\varepsilon}
,\psi_{\varepsilon})=(\omega_{0},\chi_{0}
,\psi_{0})
+\varphi_{\varepsilon}(\omega-\omega_{0},\chi-\chi_{0}
,\psi-\psi_{0}),
\end{equation}
so that $\Psi_{\varepsilon}=(\sigma,\omega_0,\chi_0,\psi_0)$ on $\mathbb{S}^{2}\setminus \Omega_{\varepsilon}$.

\begin{lemma}\label{lemma2}
$\lim_{\varepsilon\rightarrow 0}\mathcal{I}_{a}(\Psi_{\varepsilon})=\mathcal{I}_{a}(\Psi).$
\end{lemma}

\begin{proof}
Write
\begin{equation}\label{69}
\mathcal{I}_{a}(\Psi_{\varepsilon})
=\mathcal{I}_{a}(\Psi_{\varepsilon})|_{\sin\theta\leq\varepsilon}
+\mathcal{I}_{a}(\Psi_{\varepsilon})|_{\varepsilon<\sin\theta<\sqrt{\varepsilon}}
+\mathcal{I}_{a}(\Psi_{\varepsilon})|_{\sin\theta\geq\sqrt{\varepsilon}},
\end{equation}
and observe that
\begin{equation}\label{69.1}
\mathcal{I}_{a}(\Psi_{\varepsilon})|_{\sin\theta\geq\sqrt{\varepsilon}}
\rightarrow \mathcal{I}_{a}(\Psi)
\end{equation}
by the dominated convergence theorem. Moreover
\begin{align}\label{70}
\begin{split}
\mathcal{I}_{a}(\Psi_{\varepsilon})|_{\sin\theta\leq\varepsilon}
=&\int_{\sin\theta\leq\varepsilon}\left(\underbrace{(1+\Lambda a^2\cos^2\theta)\sigma}_{O(1)} +\underbrace{\zeta_{a}\Lambda\left(\frac{A}{4\pi}\right)^2 e^{-\sigma}}_{O(1)}\right)dA\\
&
+\int_{\sin\theta\leq\varepsilon}
\zeta_{a}\left(\underbrace{\frac{\sigma'^2}{4}}_{O(1)}
+\frac{e^{-2\sigma}}{\sin^4\theta}
\underbrace{(\omega_0'+\chi_0\psi_0'-\psi_0\chi_0')^2}_{O(\sin^{6}\theta)}
+\frac{e^{-\sigma}}{\sin^2\theta}\underbrace{(\chi_{0}'^{2}+\psi_{0}'^{2})}_{O(\sin^{2}\theta)}\right)dA\\
\rightarrow & 0.
\end{split}
\end{align}

Next consider the region $\varepsilon<\sin\theta<\sqrt{\varepsilon}$, and note that uniform boundedness of the following integrand implies
\begin{equation}\label{71}
\int_{\varepsilon<\sin\theta<\sqrt{\varepsilon}}\left((1+\Lambda a^2\cos^2\theta)\sigma +\zeta_{a}\Lambda\left(\frac{A}{4\pi}\right)^2 e^{-\sigma}\right)dA=O(\sqrt{\varepsilon}).
\end{equation}
To proceed further use the fact that $\omega|_{\Gamma}=\omega_{0}|_{\Gamma}$, $\chi|_{\Gamma}=\chi_{0}|_{\Gamma}$, $\psi|_{\Gamma}=\psi_{0}|_{\Gamma}$ together with
\eqref{65} and \eqref{66} yields
\begin{equation}\label{72}
|\omega-\omega_0|+|\chi-\chi_0|+|\psi-\psi_0|=O(\sin^{2}\theta).
\end{equation}
Then since
\begin{equation}\label{73}
\chi_{\varepsilon}'=\varphi_{\varepsilon}\chi'+(1-\varphi_{\varepsilon})\chi_{0}'
+\varphi_{\varepsilon}'(\chi-\chi_0)
\end{equation}
and similarly for $\psi_{\varepsilon}'$, we have that for some constant $C$ independent of $\varepsilon$
\begin{align}\label{74}
\begin{split}
&\int_{\varepsilon<\sin\theta<\sqrt{\varepsilon}}
\zeta_{a}\frac{e^{-\sigma}}{\sin^{2}\theta}(\chi_{\varepsilon}'^2+\psi_{\varepsilon}'^2)dA\\
\leq &\int_{\varepsilon<\sin\theta<\sqrt{\varepsilon}}\frac{C}{\sin\theta}\left(
\underbrace{\chi'^2}_{O(\sin^2\theta)}+\underbrace{\chi_{0}'^2}_{O(\sin^2\theta)}
+\underbrace{\varphi_{\varepsilon}'^2(\chi-\chi_0)^2}_{O\left(\frac{\sin\theta}{\log\varepsilon}\right)^2}
+\underbrace{\psi'^2}_{O(\sin^2\theta)}+\underbrace{\psi_{0}'^2}_{O(\sin^2\theta)}
+\underbrace{\varphi_{\varepsilon}'^2
(\psi-\psi_0)^2}_{O\left(\frac{\sin\theta}{\log\varepsilon}\right)^2}\right)d\theta\\
=&O(\varepsilon).
\end{split}
\end{align}
Finally a calculation shows that
\begin{align}\label{75}
\begin{split}
\omega_{\varepsilon}'+\chi_{\varepsilon}\psi_{\varepsilon}'-\psi_{\varepsilon}\chi_{\varepsilon}'
=&\varphi_{\varepsilon}(\omega'+\chi\psi'-\psi\chi')
+(1-\varphi_{\varepsilon})(\omega_{0}'+\chi_{0}\psi_{0}'-\psi_{0}\chi_{0}')
+\varphi_{\varepsilon}'(\omega-\omega_{0})\\
&+\varphi_{\varepsilon}'(\chi_{0}\psi-\psi_{0}\chi)
+\varphi_{\varepsilon}(1-\varphi_{\varepsilon})[(\psi-\psi_0)(\chi-\chi_{0})'
-(\chi-\chi_0)(\psi-\psi_0)'],
\end{split}
\end{align}
and hence
\begin{align}\label{76}
\begin{split}
&\int_{\varepsilon<\sin\theta<\sqrt{\varepsilon}}\zeta_{a}\frac{e^{-2\sigma}}{\sin^4\theta}
(\omega_{\varepsilon}'+\chi_{\varepsilon}\psi_{\varepsilon}'-\psi_{\varepsilon}\chi_{\varepsilon}')^2 dA\\
\leq &\int_{\varepsilon<\sin\theta<\sqrt{\varepsilon}}\frac{C}{\sin^3\theta}
\left(\underbrace{(\omega'+\chi\psi'-\psi\chi')^2}_{O(\sin^{2+2\delta}\theta)}
+\underbrace{(\omega_0'+\chi_0\psi_0'-\psi_0\chi_0')^2}_{O(\sin^{6}\theta)}
+\underbrace{\varphi_{\varepsilon}'^2(\omega-\omega_0)^2}
_{O\left(\frac{\sin\theta}{\log\varepsilon}\right)^2}\right)d\theta\\
&+\int_{\varepsilon<\sin\theta<\sqrt{\varepsilon}}\frac{C}{\sin^3\theta}
\left(\underbrace{\varphi_{\varepsilon}'^2(\chi_0\psi-\psi_0\chi)^2}
_{O\left(\frac{\sin\theta}{\log\varepsilon}\right)^2}
+\underbrace{(\psi-\psi_0)^2(\chi'-\chi_0')^2}_{O(\sin^6\theta)}
+\underbrace{(\chi-\chi_0)^2(\psi'-\psi_0')^2}_{O(\sin^6\theta)}\right)d\theta\\
=& O\left(\frac{1}{|\log\varepsilon|}\right).
\end{split}
\end{align}
It follows that $\mathcal{I}_{a}(\Psi_{\varepsilon})|_{\varepsilon<\sin\theta<\sqrt{\varepsilon}}\rightarrow 0$.
\end{proof}

The next proposition establishes the primary tool used in the minimization procedure. Namely, the quasi-harmonic energy \eqref{57} is convex along geodesic deformations.

\begin{prop}\label{prop1}
Let $F_{t}:\Omega\rightarrow\mathbb{H}_{\mathbb{C}}^2$ be a family of smooth maps, where $\Omega$ is a domain in $\mathbb{S}^2$. Suppose that for each $(\theta,\phi)\in\Omega$ the curve $t\mapsto F_{t}(\theta,\phi)$, $t\in[0,1]$ is a geodesic, then
\begin{equation}\label{76.1}
\frac{d^2}{dt^2} E_{\Omega}(F_{t})\geq
2\int_{\Omega}|\nabla\operatorname{dist}_{\mathbb{H}_{\mathbb{C}}^2}(F_{1},F_{0})|^2dA.
\end{equation}
\end{prop}

\begin{proof}
Let $\gamma_0$ and $h_0$ denote the metrics on the round 2-sphere and complex hyperbolic plane, respectively. Then the square of the harmonic energy density is
\begin{equation}\label{77}
|dF_{t}|^2=\gamma_{0}^{ij}(h_{0})_{lm}\partial_{i}F_{t}^{l}\partial_{j}F_{t}^{m}.
\end{equation}
We then have
\begin{equation}\label{78}
E_{\Omega}(F_{t})=\int_{\Omega}\zeta_{a}\left(|dF_{t}|^2+\Lambda\left(\frac{A}{4\pi}\right)^2
e^{2u_t}\sin^2\theta\right)dA,
\end{equation}
where $F_{t}=(u_t,\omega_t,\chi_t,\psi_t)$. Observe that since $F_t$ is a geodesic and $\mathbb{H}_{\mathbb{C}}^2$ is negatively curved
\begin{align}\label{79}
\begin{split}
\frac{d^2}{dt^2}\frac{1}{2}\int_{\Omega}\zeta_{a}|dF_{t}|^2dA
=&\int_{\Omega}\zeta_a \gamma_0^{ij}\left(\langle\nabla_{t}\partial_{i}F_{t},\nabla_{t}\partial_{j}F_{t}\rangle_{h_{0}}
+\langle\nabla_{t}\nabla_{i}\partial_{t}F_{t},\partial_{j}F_{t}\rangle_{h_{0}}\right)dA\\
=&\int_{\Omega}\zeta_a \gamma_0^{ij}\left(\langle\nabla_{i}\partial_{t}F_{t},\nabla_{j}\partial_{t}F_{t}\rangle_{h_{0}}
+\langle R^{\mathbb{H}_{\mathbb{C}}^2}(\partial_{t}F_{t},\partial_{i}F_{t})\partial_{t}F_{t},
\partial_{j}F_{t}\rangle_{h_{0}}\right)dA\\
\geq &\int_{\Omega}\zeta_a |\nabla|\partial_{t}F_{t}|_{h_0}|_{\gamma_0}^2 dA\\
\geq&\int_{\Omega}|\nabla\operatorname{dist}_{\mathbb{H}_{\mathbb{C}}^2}(F_{1},F_{0})|^2dA,
\end{split}
\end{align}
where the last step follows from the fact that $\zeta_{a}\geq 1$ and $F_t$ is a geodesic
parameterized on the interval $[0,1]$, so that $|\partial_{t}F_{t}|_{h_0}
=\operatorname{dist}_{\mathbb{H}_{\mathbb{C}}^2}(F_{1},F_{0})$.

In what follows, for simplicity of notation, we refrain from indicating dependence on $t$. To complete the proof it is sufficient to show that
\begin{equation}\label{80}
\partial_{t}^{2}e^{2u}=2(\ddot{u}+2\dot{u}^2)e^{2u}\geq 0,
\end{equation}
where $\dot{u}=\partial_{t}u$. From the geodesic equation we have
\begin{equation}\label{81}
\ddot{u}+\Gamma_{jl}^{u}\dot{F}^{j}\dot{F}^{l}=0,
\end{equation}
and a computation of Christoffel symbols for $h_0$ yields
\begin{equation}\label{82}
\Gamma_{uu}^{u}=\Gamma_{u\omega}^{u}=\Gamma_{u\chi}^{u}=\Gamma_{u\psi}^{u}=0,\quad
\Gamma_{\omega\omega}^{u}=-2e^{4u},\quad \Gamma_{\omega\chi}^{u}=2\psi e^{4u},\quad
\Gamma_{\omega\psi}^{u}=-2\chi e^{4u},
\end{equation}
\begin{equation}\label{83}
\Gamma_{\chi\chi}^{u}=-e^{2u}-2\psi^2 e^{4u},\quad
\Gamma_{\psi\psi}^{u}=-e^{2u}-2\chi^2 e^{4u},\quad
\Gamma_{\chi\psi}^{u}=2\chi\psi e^{4u}.
\end{equation}
It follows that
\begin{equation}\label{84}
\ddot{u}=2 e^{4u}\dot{\omega}^2
+(e^{2u}+2\psi^2 e^{4u})\dot{\chi}^2+(e^{2u}+2\chi^2 e^{4u})\dot{\psi}^2
-4\psi e^{4u}\dot{\omega}\dot{\chi}+4\chi e^{4u}\dot{\omega}\dot{\psi}
-4\chi\psi e^{4u}\dot{\chi}\dot{\psi}.
\end{equation}
Note that the expression on the right-hand side is related to the constant squared length of the velocity vector for the geodesic
\begin{equation}\label{85}
|\dot{F}|_{h_0}^2=\dot{u}^2+e^{4u}(\dot{\omega}+\chi\dot{\psi}-\psi\dot{\chi})^2
+e^{2u}(\dot{\chi}^2+\dot{\psi}^2).
\end{equation}
Therefore
\begin{equation}\label{86}
\ddot{u}+2\dot{u}^2
=2|\dot{F}|_{h_0}^2-e^{2u}(\dot{\chi}^2+\dot{\psi}^2)\geq 0.
\end{equation}
\end{proof}

We are now in a position to prove the main result of this section.\medskip

\noindent\textit{Proof of Theorem \ref{thm2}.} Recall that $\Psi_{\varepsilon}$ satisfies
\eqref{62}. Thus, if $\tilde{\Psi}^{t}_{\varepsilon}$ is the geodesic connecting $\tilde{\Psi}_{0}$ to
$\tilde{\Psi}_{\varepsilon}$ as described at the beginning of this section, then
$\sigma^{t}=\sigma_{0}+t(\sigma-\sigma_{0})$. Observe that
\begin{equation}\label{87}
\frac{d^{2}}{dt^{2}}\mathcal{I}_a(\Psi^{t}_{\varepsilon})
=\underbrace{\frac{d^{2}}{dt^{2}}\mathcal{I}_{\Omega_{\varepsilon}}(\Psi^{t}_{\varepsilon})}_{I_{1}}+
\underbrace{\frac{d^{2}}{dt^{2}}\mathcal{I}_{\mathbb{S}^2\setminus\Omega_{\varepsilon}}
(\Psi^{t}_{\varepsilon})}_{I_{2}}.
\end{equation}
Then using the linearity of $\sigma^t$ together with Proposition
\ref{prop1} produces
\begin{align}\label{88}
\begin{split}
I_{1}&=\frac{d^{2}}{dt^{2}}E_{\Omega_{\varepsilon}}
(\tilde{\Psi}^{t}_{\varepsilon})
-\frac{d^{2}}{dt^{2}}\int_{\partial\Omega_{\varepsilon}}
\zeta_{a}(\sigma_{0}+t(\sigma-\sigma_{0})+2\log\sin\theta)\partial_{\nu}\log\sin\theta ds\\
&\geq 2\int_{\Omega_{\varepsilon}}
|\nabla\operatorname{dist}_{\mathbb{H}^{2}_{\mathbb{C}}}
(\tilde{\Psi}_{\varepsilon},\tilde{\Psi}_{0})|^{2}dA
\end{split}
\end{align}
On the other hand, direct computation yields
\begin{align}\label{89}
\begin{split}
I_2=&\int_{\mathbb{S}^2\setminus\Omega_{\varepsilon}}
\zeta_{a}\left(\frac{(\sigma-\sigma_0)'^2}{2}
+\frac{4(\sigma-\sigma_0)^2 e^{-2\sigma^t}}{\sin^4\theta}
(\omega_{0}'+\chi_0\psi_0'-\psi_0\chi_0')^2\right)dA\\
&+\int_{\mathbb{S}^2\setminus\Omega_{\varepsilon}}
\zeta_{a}\left(\frac{(\sigma-\sigma_0)^2 e^{-\sigma^t}}{\sin^2\theta}(\chi_0'^2+\psi_0'^2)
+\Lambda\left(\frac{A}{4\pi}\right)^2(\sigma-\sigma_0)^2 e^{-\sigma^t}\right)dA\\
\geq & \int_{\mathbb{S}^2\setminus\Omega_{\varepsilon}}\frac{(\sigma-\sigma_0)'^2}{2}dA\\
=&2\int_{\mathbb{S}^2\setminus\Omega_{\varepsilon}}
|\nabla\operatorname{dist}_{\mathbb{H}^{2}_{\mathbb{C}}}
(\tilde{\Psi}_{\varepsilon},\tilde{\Psi}_{0})|^{2}dA,
\end{split}
\end{align}
since $\operatorname{dist}_{\mathbb{H}^{2}_{\mathbb{C}}}
(\tilde{\Psi}_{\varepsilon},\tilde{\Psi}_{0})=|u-u_0|$ on $\mathbb{S}^2\setminus\Omega_{\varepsilon}$. Note that the passing of $\frac{d^{2}}{dt^{2}}$ into the integral in \eqref{89} is justified, since each term on the right-hand side of the first equality is uniformly integrable. Combining \eqref{88} and \eqref{89} gives the desired convexity statement
\begin{equation}\label{87}
\frac{d^{2}}{dt^{2}}\mathcal{I}_a(\Psi^{t}_{\varepsilon})
\geq 2\int_{\mathbb{S}^2}
|\nabla\operatorname{dist}_{\mathbb{H}^{2}_{\mathbb{C}}}
(\tilde{\Psi}_{\varepsilon},\tilde{\Psi}_{0})|^{2}dA.
\end{equation}

We next observe that \eqref{64} holds. To see this, use that the extreme KNdS map $\Psi_0$ satisfies the Euler-Lagrange equations for the functional $\mathcal{I}_{a}$ (see Appendix \ref{sec5}), together with the fact that $\frac{d}{dt}\omega^{t}_{\varepsilon}=
\frac{d}{dt}\chi^{t}_{\varepsilon}=\frac{d}{dt}\psi^{t}_{\varepsilon}=0$ in a neighborhood of the axis $\Gamma$, to find
\begin{equation}\label{88}
\frac{d}{dt}\mathcal{I}_{a}(\Psi_{t})|_{t=0}
=\frac{1}{2}\zeta_{a}\sigma_0'(\sigma-\sigma_0)\sin\theta|_{\theta=0}^{\pi}=0.
\end{equation}
Note that justification for passing $\frac{d}{dt}$ into the integral is analogous to that in the
previous paragraph.
Now integrating \eqref{87} twice and applying the Poincar\'{e} inequality produces
\begin{equation}\label{89}
\mathcal{I}_{a}(\Psi_{\varepsilon})-\mathcal{I}_{a}(\Psi_{0})
\geq 2\int_{\mathbb{S}^{2}}|\nabla\operatorname{dist}_{\mathbb{H}^{2}_{\mathbb{C}}}
(\tilde{\Psi}_{\varepsilon},\tilde{\Psi}_{0})|^{2}dA
\geq C\int_{\mathbb{S}^{2}}\left(
\operatorname{dist}_{\mathbb{H}_{\mathbb{C}}^{2}}
(\tilde{\Psi}_{\varepsilon},\tilde{\Psi}_{0})-D_{\varepsilon}\right)^2 dA,
\end{equation}
where $D_{\varepsilon}$ is the average value of $\operatorname{dist}_{\mathbb{H}_{\mathbb{C}}^{2}}
(\tilde{\Psi}_{\varepsilon},\tilde{\Psi}_{0})$.

By Lemma \ref{lemma2} $\lim_{\varepsilon\rightarrow 0}\mathcal{I}_{a}(\Psi_{\varepsilon})=\mathcal{I}_{a}(\Psi)$, and thus in order to complete the
proof it suffices to show that the limit may be passed under the integral on the right-hand side. By the triangle inequality and some algebra, it is enough to show
\begin{equation}\label{90}
\lim_{\varepsilon\rightarrow 0}\int_{\mathbb{S}^{2}}
\operatorname{dist}_{\mathbb{H}_{\mathbb{C}}^{2}}^{2}
(\tilde{\Psi}_{\varepsilon},\tilde{\Psi})dA
=0.
\end{equation}
To see this, use the
triangle inequality and direct calculation to find
\begin{align}\label{91}
\begin{split}
\operatorname{dist}_{\mathbb{H}_{\mathbb{C}}^{2}}
(\tilde{\Psi}_{\varepsilon},\tilde{\Psi})
\leq&
\operatorname{dist}_{\mathbb{H}_{\mathbb{C}}^{2}}
((u,\omega_{\varepsilon},\chi_{\varepsilon},
\psi_{\varepsilon}),(u,\omega,\chi_{\varepsilon},
\psi_{\varepsilon}))+\operatorname{dist}_{\mathbb{H}_{\mathbb{C}}^{2}}
((u,\omega,\chi_{\varepsilon},
\psi_{\varepsilon}),(u,\omega,\chi,
\psi_{\varepsilon}))\\
&
+\operatorname{dist}_{\mathbb{H}_{\mathbb{C}}^{2}}
((u,\omega,\chi,
\psi_{\varepsilon}),(u,\omega,\chi,
\psi))\\
\leq& C\left[e^{2u}(|\omega-\omega_{\varepsilon}|
+|\psi_{\varepsilon}||\chi-\chi_{\varepsilon}|
+|\chi||\psi-\psi_{\varepsilon}|)
+e^{u}(|\chi-\chi_{\varepsilon}|
+|\psi-\psi_{\varepsilon}|)\right].
\end{split}
\end{align}
Since the right-hand side is uniformly bounded independent of $\varepsilon$, the dominated
convergence theorem applies to give \eqref{90}.\hfill\qedsymbol\medskip

\textit{Proof of Theorem \ref{thm1}.}
From the given initial data $(M,g,k,E,B)$ we obtain the four quantities $(\sigma,\omega,\chi,\psi)$ consisting of a metric component and three potentials, as explained in Section \ref{sec2}. Let $(A,\mathcal{J},Q)$ be the area, angular momentum, and charge of the horizon $S\subset M$. From Lemma \ref{scaling} there exists a corresponding triple $(\hat{A},
\hat{\mathcal{J}},\hat{Q})$ which arises from an extreme KNdS solution, and is such that the desired inequality \eqref{9} is reduced to showing $\hat{A}\geq A$. Let
\begin{equation}\label{92}
\hat{a}=\frac{\hat{\mathcal{J}}}{m(\hat{\mathcal{J}},\hat{Q})},\quad\quad\quad
\hat{\Psi}=(\hat{\sigma},\hat{\omega},\hat{\chi},\hat{\psi})
=\left(\sigma+\log\frac{\hat{A}}{A},\frac{\hat{A}^2}{A^2}\omega,\frac{\hat{A}}{A}\chi,
\frac{\hat{A}}{A}\psi\right),
\end{equation}
then Lemma \ref{lemma3} asserts that $\hat{A}\geq A$ is valid as long as
\begin{equation}\label{93}
\mathcal{I}_{\hat{a}}(\hat{\Psi})\geq\mathcal{I}_{\hat{a}}(\Psi_0),
\end{equation}
where $\Phi_0$ is the extreme KNdS map with the same angular momentum and charge $(\hat{\mathcal{J}},\hat{Q})$. Finally, observe that Theorem \ref{thm2} is applicable, since
smoothness of the initial data together with the potential formulas \eqref{24} and \eqref{27} guarantee that the asymptotics \eqref{66} hold for $\hat{\Psi}$. This establishes \eqref{93}
and completes the proof of inequality \eqref{9}.

Consider now the case of equality in \eqref{9}. From the proof of Lemma \ref{scaling}, this yields $(\hat{A},\hat{\mathcal{J}},\hat{Q})=(A,\mathcal{J},Q)$ and hence $\hat{\Psi}=\Psi$. In particular, the equality of areas implies that $\mathcal{I}_{a}(\Psi)\geq\mathcal{I}_{a}(\Psi_0)$
which gives $\Psi=\Psi_0$ from the gap bound in Theorem \ref{thm2}. Namely, the gap bound gives that $\operatorname{dist}_{\mathbb{H}_{\mathbb{C}}^2}(\tilde{\Psi},\tilde{\Psi}_0)$ is constant, but $\operatorname{dist}_{\mathbb{H}_{\mathbb{C}}^2}(\tilde{\Psi},\tilde{\Psi}_0)$=0 at the axis $\Gamma$, and hence it must vanish identically.  In light of \eqref{54}, equality of the areas also produces equality in the stability inequality \eqref{33}, when $v$ is  chosen as in \eqref{35}. It follows that on $S$
\begin{equation}\label{94}
\mu_{EM}-|J_{EM}(n)|=J_{EM}(\partial_{\theta})=J_{EM}(\partial_{\phi})=|II|=
\partial_{\phi}\log\varphi_{1}=|(\partial_{\theta}\log\varphi_1-X(\partial_{\theta}))v
-\partial_{\theta}v|=0,
\end{equation}
and
\begin{equation}\label{95}
|E|^2+|B|^2-E(n)^2-B(n)^2-2E\times B(n)=0.
\end{equation}
A computation shows that \eqref{95} implies
\begin{equation}\label{96}
E(e_2)=B(e_3),\quad\quad\quad E(e_3)=-B(e_2),
\end{equation}
and since the proofs above are invariant under the transformation $E\mapsto -E$ we must then have
\begin{equation}\label{97}
E(e_2)=E(e_3)=B(e_2)=B(e_3)=0.
\end{equation}
Furthermore, the potential formulas \eqref{24} show that $E(n)$ and $B(n)$ agree with their counterparts in the extreme KNdS solution, and thus the full electromagnetic field $(E,B)$ is that of the extreme KNdS solution on the horizon.

Lastly, from \eqref{94} we have
\begin{equation}\label{98}
k(n,\partial_{\theta})=\partial_{\theta}\log\left(\frac{\varphi_1}{v}\right)
=\partial_{\theta}\log\left(\frac{e^{\sigma_{0}/2}\varphi_1}{\sqrt{\zeta_{a}}}\right).
\end{equation}
Moreover, the potential equation \eqref{27} implies that $k(n,\partial_{\phi})$ equates with its counterpart in the extreme KNdS spacetime. All together this shows that the coefficients of the stability operator $L$ arise from the extreme KNdS data, so that the eigenfunction satisfying $L\varphi_1=0$ corresponds to the same quantity in the extreme KNdS setting. Hence $k(n,\partial_{\theta})$ agrees with its counterpart in the extreme KNdS solution. We conclude that $(S,\gamma,k(n,\cdot),E,B)$ arises from an extreme KNdS horizon.
\hfill\qedsymbol\medskip

\appendix
\numberwithin{equation}{section}

\section{The Kerr-Newman-de Sitter Spacetime}\label{sec5}

The Kerr-Newman-de Sitter black hole solves the Einstein-Maxwell equations with positive cosmological constant
\begin{equation}\label{5.1}
\tilde{R}_{ab}-\frac{1}{2}\tilde{R}\tilde{g}_{ab}+\Lambda\tilde{g}_{ab}
=8\pi T_{ab}=2\left(F_{ac}F_{b}^{\phantom{b}c}-\frac{1}{4}|F|^2\tilde{g}_{ab}\right),
\end{equation}
\begin{equation}\label{5.2}
dF=0,\quad\quad\quad d\star_{4} F=0.
\end{equation}
In Boyer-Lindquist-like coordinates the KNdS metric \cite{CaldarelliCognolaKlemm} is given by
\begin{equation}\label{5.3}
\tilde{g}=-\frac{\Delta_r}{\Sigma}\left(dt-\frac{a\sin^2\theta}{\Xi}d\phi\right)^2
+\frac{\Sigma}{\Delta_r}dr^2+\frac{\Sigma}{\Delta_{\theta}}d\theta^2
+\frac{\Delta_{\theta}\sin^2\theta}{\Sigma}\left(adt-\frac{r^2+a^2}{\Xi}d\phi\right)^2,
\end{equation}
where
\begin{equation}\label{5.4}
\Delta_r=(r^2+a^2)\left(1-\frac{r^2 \Lambda}{3}\right)-2mr+q^2,\quad\quad\quad
\Xi=1+\frac{a^2 \Lambda}{3},
\end{equation}
\begin{equation}\label{5.5}
\Delta_{\theta}=1+\frac{a^2 \Lambda}{3}\cos^2\theta,\quad\quad\quad
\Sigma=r^2+a^2 \cos^2\theta,
\end{equation}
and the field strength and vector potential ($F=d\mathcal{A}$) take the form
\begin{equation}\label{5.6}
\mathcal{A}=\frac{q_e r}{\sqrt{\Delta_r \Sigma}}e^0 - \frac{q_b \cos\theta}
{\sqrt{\Delta_{\theta}\Sigma}\sin\theta} e^3,
\end{equation}
\begin{equation}\label{5.7}
F=\frac{1}{\Sigma^2}\left[\left(q_e(r^2-a^2 \cos^2\theta)-2q_b r a\cos\theta\right) e^0 \wedge e^1
\left(q_b(r^2 -a^2\cos^2\theta)+2q_e r a\cos\theta\right)e^2\wedge e^3\right].
\end{equation}
Here the following orthonormal coframe is used
\begin{equation}\label{5.8}
e^0=\sqrt{\frac{\Delta_r}{\Sigma}}\left(dt-\frac{a\sin^2\theta}{\Xi}d\phi\right),\quad\quad
e^1=\sqrt{\frac{\Sigma}{\Delta_r}}dr,
\end{equation}
\begin{equation}\label{5.8.1}
e^2=\sqrt{\frac{\Sigma}{\Delta_{\theta}}}d\theta,\quad\quad
e^3=\sqrt{\frac{\Delta_{\theta}}{\Sigma}}\sin\theta\left(adt-\frac{r^2+a^2}{\Xi}d\phi\right),
\end{equation}
and the parameters $m$, $a$, and $q=\sqrt{q_e^2+q_b^2}$ define the mass, angular momentum, and charge through the formulas
\begin{equation}\label{5.9}
\mathcal{M}=\frac{m}{\Xi^2},\quad\quad\mathcal{J}=\frac{am}{\Xi^2},\quad\quad
Q_e=\frac{q_e}{\Xi},\quad\quad Q_b=\frac{q_b}{\Xi}.
\end{equation}

The geometry of the KNdS solution depends crucially on the zeros of the polynomial $\Delta_r$. In order to avoid naked singularities and other undesirable features, the relevant parameters must satisfy certain restrictions. If
\begin{equation}\label{5.10.0}
a^2 \Lambda<3\quad\quad\text{ and }\quad\quad m_{crit}^{-}\leq m\leq m_{crit}^{+},
\end{equation}
where $m_{crit}^{\pm}$ are the two positive solutions of the equation
\begin{align}\label{5.10}
\begin{split}
0=&m^4+\frac{(a^2\Lambda-3)\left((a^2\Lambda-3)^2+108\Lambda(a^2+q^2)\right)}{3^5\Lambda}m^2\\
&+\frac{(a^2+q^2)\left((a^2\Lambda-3)^2+12\Lambda(a^2+q^2)\right)^2}{3^6\Lambda},
\end{split}
\end{align}
then $\Delta_r$ has four real roots $r_{--}<r_{-}\leq r_{+}\leq r_{c}$, one of which is simple and negative with the rest positive. The roots $r_{-}$ and $r_{+}$ represent inner and outer event horizons, while the root $r_{c}$ corresponds to a de Sitter cosmological horizon. An extremal black hole occurs when at least two of the three positive roots coincide, and in this situation the geometry near the horizon becomes asymptotically cylindrical. In particular, if $m=m_{crit}^{-}$ then $r_{-}=r_{+}$ and if $m=m_{crit}^{+}$ then $r_{+}=r_{c}$.

We now derive the quasi-harmonic map $\Psi_0=(\sigma_0,\omega_0,\chi_0,\psi_0)$ associated with
an extreme KNdS solution. At an extreme horizon $\Delta_r=0$ and thus the induced metric is given by
\begin{equation}\label{5.11}
\gamma_0=\frac{\Sigma}{\Delta_{\theta}}d\theta^2+\frac{\Delta_{\theta}(r_{+}^2+a^2)^2\sin^2\theta}
{\Sigma\Xi^2}d\phi^2.
\end{equation}
This easily fits into the canonical form \eqref{23} by setting
\begin{equation}\label{5.12}
\sigma_0=\log\frac{\Delta_{\theta}(r_{+}^2+a^2)}{\Xi^2\Sigma},\quad\quad\quad
\frac{A}{4\pi}=e^c=\frac{r_{+}^2+a^2}{\Xi}.
\end{equation}
To find the electromagnetic potentials recall from \eqref{24} that
\begin{equation}\label{5.12.1}
\chi_{0}'=E(e_1)e^c \sin\theta,\quad\quad\quad \psi_{0}'=B(e_1)e^c \sin\theta,
\end{equation}
where $\{e_i\}$, $i=0,1,2,3$ form the frame dual to \eqref{5.8}, \eqref{5.8.1}.
Moreover
\begin{equation}\label{5.13}
E=\iota_{e_{0}}F=E(e_1)e^1,\quad\quad\quad B=\iota_{e_{0}}\star_{4}F=B(e_1)e^1,
\end{equation}
with
\begin{equation}\label{5.14}
E(e_1)=\frac{1}{\Sigma^2}\left[q_e(r_{+}^2-a^2\cos^2\theta)-2q_b r_{+}a\cos\theta\right],
\end{equation}
\begin{equation}\label{5.14.1}
B(e_1)=\frac{1}{\Sigma^2}\left[q_b(r_{+}^2-a^2\cos^2\theta)+2q_e r_{+}a\cos\theta\right].
\end{equation}
It then follows that
\begin{equation}\label{5.15}
\chi_0=-\frac{1}{\Xi\Sigma}\left[q_e (r_{+}^2+a^2)\cos\theta+q_b r_{+} a\sin^2\theta\right],
\end{equation}
\begin{equation}\label{5.15.1}
\psi_0=-\frac{1}{\Xi\Sigma}\left[q_b (r_{+}^2+a^2)\cos\theta-q_e r_{+} a\sin^2\theta\right].
\end{equation}

In order to find the charged twist potential recall from \eqref{27} that
\begin{equation}\label{5.16}
\omega_0'=k(e_1,\partial_{\phi})e^c \sin\theta-\chi_0\psi_0'+\psi_0\chi_0',
\end{equation}
where the second fundamental form may be expressed as
\begin{equation}\label{5.17}
k(e_1,\partial_{\phi})=\langle \partial_{\phi},\nabla_{e_1}N\rangle
=-\langle \nabla_{e_1}\partial_{\phi},N\rangle=-e^0(\nabla_{e_1}\partial_{\phi})
\end{equation}
in which $N$ is the unit normal to the $t=0$ slice. A computation with Christoffel symbols then yields
\begin{equation}\label{5.18}
e^0(\nabla_{e_1}\partial_{\phi})=\frac{\Delta_{r}}{\Sigma}\left(\Gamma_{r\phi}^{t}
-\frac{a\sin^2\theta}{\Xi}\Gamma_{r\phi}^{\phi}\right)
=-\frac{ar_{+}(r_{+}^2+a^2)\Delta_{\theta}\sin^2\theta}{\Xi}.
\end{equation}
Furthermore
\begin{equation}\label{5.19}
-\chi_0\chi_0'+\psi_0\chi_0'
=\frac{q^2 a r_{+}(r_{+}^2+a^2)}{\Xi^2\Sigma^2}\sin\theta(1+\cos^2\theta),
\end{equation}
and since $\Delta_r=0$ the following relation holds
\begin{equation}\label{5.20}
(r_{+}^2+a^2)\Delta_{\theta}=(r_{+}^2+a^2)\left(1-\frac{r_{+}^2\Lambda}{3}\right)
+\frac{\Lambda}{3} (r_{+}^2+a^2)\Sigma=2mr_{+}-q^2+\frac{\Lambda}{3} (r_{+}^2+a^2)\Sigma.
\end{equation}
By combining \eqref{5.18}, \eqref{5.19}, and \eqref{5.20} we arrive at
\begin{equation}\label{5.21}
\omega_{0}'
=\frac{ar_{+}(r_{+}^2+a^2)}{\Xi^2\Sigma^2}
\left[\left(2mr_{+}+\frac{\Lambda}{3}(r_{+}^2+a^2)\Sigma\right)\sin^2\theta
+2q^2\cos^2\theta\right].
\end{equation}
Integration then produces
\begin{align}\label{5.22}
\begin{split}
\omega_0=&\frac{ar_{+}(r_{+}^2+a^2)}{\Xi^2}\left[\frac{\cos\theta}{\Sigma}
\left(-\frac{m(r_{+}^2+a^2)}{a^2 r_{+}}+\frac{q^2}{a^2}
+\frac{\Lambda(r_{+}^2+a^2)}{3a^2}\Sigma\right)\right.\\
&+\left.\arctan\left(\frac{a\cos\theta}{r_{+}}\right)
\left(\frac{m(r_{+}^2-a^2)}{a^3 r_{+}^2}
-\frac{q^2}{a^3 r_{+}}-\frac{\Lambda(r_{+}^2+a^2)^2}{3 a^3 r_{+}}\right)\right].
\end{split}
\end{align}
Note that in the case of an extreme Kerr-Newman horizon, that is $\Lambda=0$, this simplifies so that the second line vanishes and the entire expression reduces to the formula in
\cite[Lemma 3.4]{ClementJaramilloReiris}. Having constructed the extreme KNdS map $\Psi_0$, it may be verified that the values of the angular momentum and charge as given by \eqref{25}, \eqref{26}, and \eqref{28} coincide with those given in \eqref{5.9}.

The map $\Psi_0$ satisfies the Euler-Lagrange equations for the functional $\mathcal{I}_{a}$ in \eqref{47}:
\begin{align}\label{5.23}
\begin{split}
\frac{1}{\sin\theta}\left(\zeta_{a}\sin\theta\sigma'\right)'=&-2\zeta_{a}\left(\frac{2e^{-2\sigma}}{\sin^4\theta}
(\omega'+\chi\psi'-\psi\chi')^2+\frac{e^{-\sigma}}{\sin^2\theta}(\chi'^2+\psi'^2)\right)\\
&+2(1+a^2 \Lambda\cos^2\theta)-2\Lambda\zeta_a\left(\frac{A}{4\pi}\right)^2 e^{-\sigma},
\end{split}
\end{align}
\begin{equation}\label{5.24}
\left(\zeta_a \frac{e^{-2\sigma}}{\sin^4\theta}(\omega'+\chi\psi'-\psi\chi')\right)'=0,
\end{equation}
\begin{equation}\label{5.25}
\frac{1}{\sin\theta}\left(\zeta_{a}\frac{e^{-\sigma}}{\sin\theta}\chi'\right)'
-\frac{2\zeta_a e^{-2\sigma}}{\sin^4\theta}(\omega'+\chi\psi'-\psi\chi')\psi'=0,
\end{equation}
\begin{equation}\label{5.26}
\frac{1}{\sin\theta}\left(\zeta_{a}\frac{e^{-\sigma}}{\sin\theta}\psi'\right)'
+\frac{2\zeta_a e^{-2\sigma}}{\sin^4\theta}(\omega'+\chi\psi'-\psi\chi')\chi'=0.
\end{equation}
In order to elucidate the quasi-harmonic map structure of the equations, let 
$u_{0}=-\sigma_{0}/2-\log\sin\theta$ and $\tilde{\Psi}_{0}=(u_0,\omega_0,\chi_0,\psi_0)$. Then
$\tilde{\Psi}_{0}$ satisfies the Euler-Lagrange equations for the functional $E$ in
\eqref{57}:
\begin{equation}\label{5.27}
\frac{1}{\sin\theta}\left(\zeta_a \sin\theta u'\right)'=2\zeta_a e^{4u}(\omega'+\chi\psi'+\psi\chi')^2
+\zeta_a e^{2u}\left[(\chi'^2+\psi'^2)+\Lambda\left(\frac{A}{4\pi}\right)^2 \sin^2\theta\right],
\end{equation}
\begin{equation}\label{5.28}
\left(\zeta_a \sin\theta e^{4u}(\omega'+\chi\psi'-\psi\chi')\right)'=0,
\end{equation}
\begin{equation}\label{5.29}
\frac{1}{\sin\theta}\left(\zeta_{a}\sin\theta e^{2u}\chi'\right)'
-2\zeta_a e^{4u}(\omega'+\chi\psi'-\psi\chi')\psi'=0,
\end{equation}
\begin{equation}\label{5.30}
\frac{1}{\sin\theta}\left(\zeta_{a}\sin\theta e^{2u}\psi'\right)'
+2\zeta_a e^{4u}(\omega'+\chi\psi'-\psi\chi')\chi'=0.
\end{equation}
These clearly reduce to the axisymmetric harmonic map equations for   $\mathbb{S}^2\rightarrow\mathbb{H}_{\mathbb{C}}^2$ when $\Lambda=0$.

We now indicate the derivation of the main inequality \eqref{9}. There are two methods
for doing this. The first consists of algebraic manipulations centered on the two equations
$\Delta_r=0$ and $\partial_{r}\Delta_{r}=0$, which is carried out in \cite{ClementReirisSimon}
for the uncharged case. The second method is motivated by black hole thermodynamics. Consider the
Smarr formula \cite{DehghaniKhajehAzad} for the mass of a (not necessarily extreme) KNdS solution
\begin{equation}\label{5.31}
\mathcal{M}^2=\frac{A}{16\pi}
+\frac{\pi}{A}\left(4\mathcal{J}^2+Q^2\right)
+\frac{Q^2}{2}-\frac{\Lambda\mathcal{J}^2}{3}
-\frac{\Lambda A}{24\pi}\left(Q^2+\frac{A}{4\pi}-\frac{\Lambda A^2}{96\pi^2}\right).
\end{equation}
The temperature is computed by varying the mass function with respect to the entropy $\mathcal{S}=A/4$, and is given by
\begin{equation}\label{5.32}
T:=\frac{\partial\mathcal{M}}{\partial\mathcal{S}}
=\frac{1}{8\pi\mathcal{M}}\left[1-\frac{16\pi^2}{A^2}\left(4\mathcal{J}^2+Q^2\right)
-\frac{2\Lambda}{3}\left(Q^2+\frac{A}{2\pi}\right)+\frac{\Lambda^2 A^2}{48\pi^2}\right].
\end{equation}
Moreover, the first law of black hole thermodynamics states that
\begin{equation}\label{5.33}
d\mathcal{M}=Td\mathcal{S}+\Omega d\mathcal{J}+\Phi dQ,
\end{equation}
where $\Omega$ and $\Phi$ denote the angular velocity and electric potential respectively, and have the expressions
\begin{equation}\label{5.34}
\Omega=\frac{\partial\mathcal{M}}{\partial\mathcal{J}}=\frac{\pi\mathcal{J}}{\mathcal{M}\mathcal{S}}
\left(1-\frac{\Lambda\mathcal{S}}{3\pi}\right),\quad\quad\quad
\Phi=\frac{\partial\mathcal{M}}{\partial Q}=\frac{\pi Q}{2\mathcal{M}\mathcal{S}}
\left(Q^2+\frac{\mathcal{S}}{\pi}-\frac{\Lambda\mathcal{S}^2}{3\pi^2}\right).
\end{equation}
In the dynamical setting, Hawking's area theorem (second law of black hole thermodynamics) asserts that $\dot{\mathcal{S}}\geq 0$ where dot represents a time derivative. Thus, assuming conservation of black hole angular momentum and charge, and the physically reasonable supposition that black hole mass increases with area, then the first law \eqref{5.33} implies that the temperature is nonnegative. In conclusion, heuristic physical reasoning leads to $T\geq 0$ which is equivalent
to \eqref{9}.

Lastly, we mention that given a triple $(A,\mathcal{J},Q)$ which saturates \eqref{9} (with $A\leq 4\pi/\Lambda$) as in Lemma \ref{scaling}, we may insert these values into the Smarr formula to obtain the mass as a function of $\mathcal{J}$ and $Q$. This gives the angular momentum parameter as a function of the same quantities via the formula
\begin{equation}\label{5.35}
a=\frac{\mathcal{J}}{\mathcal{M}(\mathcal{J},Q)}.
\end{equation}
From here all the remaining parameters $m$, $q_e$, and $q_b$ of \eqref{5.9} may be determined in terms of $\mathcal{J}$ and $Q$, and we may construct a KNdS spacetime. This solution must be extreme, since according to \eqref{5.32} saturation of \eqref{9} implies that $\mathcal{M}$ (as a function of $A$) has a critical point at the given triple $(A,\mathcal{J},Q)$, and hence $m$ achieves one of the extreme values $m_{crit}^{\pm}$; a calculation shows that $a(\mathcal{J},Q)$ satisfies \eqref{5.10.0}.

\section{Canonical Coordinates}\label{sec6}

The purpose of this appendix is to establish the existence of the coordinate system introduced in Section \ref{sec2} which yields the simple expression for the horizon metric \eqref{23}. As previously mentioned, arguments for the existence of such a coordinate system have been given previously \cite{AshtekarEnglePawlowskiBroeck,DainReiris} under certain hypotheses. Here we provide a detailed proof in the context appropriate for the current results.

\begin{prop}
Let $\gamma$ be a smooth axisymmetric Riemannian metric on a topological 2-sphere $S$, and denote the associated Killing field by $\eta$. Then there exist global coordinates $(\theta,\phi)$ with $\theta\in[0,\pi]$ and $\phi\in[0,2\pi)$, such that $\eta=\partial_{\phi}$ and the metric takes the form
\begin{equation}\label{6.2}
\gamma=e^{2c-\sigma}d\theta^2+e^{\sigma}\sin^{2}\theta d\phi^2,
\end{equation}
where $\sigma\in C^{\infty}(S)$ depends only on $\theta$ and $c$ is a constant related to the area by $A=4\pi e^{c}$.
\end{prop}

\begin{proof}
Let $\{S_{i}\}\subset S$ denote the components of the zero set of $\eta$. Then each $S_i$ is of even codimension \cite[Theorem 34]{Petersen}, and is therefore a point. Moreover, according to
\cite[Theorem 40]{Petersen} we have the following relation between Euler characteristics
$2=\chi(S)=\sum_{i}\chi(S_{i})$, and hence the zero set of $\eta$ consists of exactly two points, a north and south pole $\{p_{+},p_{-}\}$.

Let $\alpha(t)$, $t\in[0,l]$ be a minimizing geodesic parameterized by arclength connecting the south to north pole, that is $\alpha(0)=p_-$, $\alpha(l)=p_+$. Let $\Phi_s$, $s\in[0,2\pi)$ denote the 1-parameter flamily of isometries associated to $\eta$, and set $\alpha_{s}(t)=\Phi_s(\alpha(t))$. Then for each $s$ the curve $t\mapsto\alpha_{s}(t)$ is a geodesic, and $\alpha_{0}(t)=\alpha_{2\pi}(t)=\alpha(t)$. By construction $(t,s)$ forms a system of global coordinates on $S\setminus\alpha$, and $\gamma(\partial_{t},\partial_{t})=1$ as well as
$\gamma(\partial_{s},\partial_{s})=|\eta|^2$. In addition, the geodesic and Killing equations imply that
\begin{equation}\label{6.3}
\partial_{t}\gamma(\partial_{t},\eta)=\gamma(\partial_{t},\nabla_{\partial_{t}}\eta)=0.
\end{equation}
Thus, since $\gamma(\partial_{t},\eta)=0$ at $p_\pm$ we must have $\gamma(\partial_{t},\partial_{s})=\gamma(\partial_{t},\eta)=0$ everywhere. It follows that
\begin{equation}\label{6.4}
\gamma=dt^2+|\eta|^2 ds^2.
\end{equation}

In order to put the metric in the form \eqref{6.2} we will make use of a potential for the volume form $\varepsilon^{(2)}$. By Cartan's formula and the fact that $\eta$ is a Killling field
\begin{equation}\label{6.5}
0=\mathfrak{L}_{\eta}\varepsilon^{(2)}=d\iota_{\eta}\varepsilon^{(2)}
+\iota_{\eta}d\varepsilon^{(2)}=d\iota_{\eta}\varepsilon^{(2)}.
\end{equation}
Since $S$ is simply connected there exists a function $f$ such that $df=\iota_{\eta}\varepsilon^{(2)}$. Note that $\eta(f)=\varepsilon^{(2)}(\eta,\eta)=0$ so that
$f$ is a function of $t$ alone. Moreover, since $|df|$ vanishes only at the north and south poles, it may be assumed that $f$ is strictly increasing. Therefore a new coordinate system may be defined by
$\tilde{\theta}=f(t)$ and $\phi=s$. Observe that
\begin{equation}\label{6.6}
f'^2=|\nabla f|^2=\varepsilon^{(2)}(\nabla f,\eta)=f'\varepsilon^{(2)}(\partial_{t},\partial_{s})
=f'|\eta|,
\end{equation}
and hence $f'=|\eta|$. It follows that
\begin{equation}\label{6.7}
\gamma=|\eta|^{-2}d\tilde{\theta}^2+|\eta|^2 d\phi^2.
\end{equation}
Now set
\begin{equation}\label{6.8}
\cos\theta=1-\frac{2(\tilde{\theta}-f(0))}{f(l)-f(0)},
\end{equation}
then
\begin{equation}\label{6.9}
\gamma=\frac{(f(l)-f(0))^2\sin^2\theta}{4|\eta|^2}d\theta^2+|\eta|^2 d\phi^2.
\end{equation}
Finally, defining $e^{-\sigma}=|\eta|^{-2}\sin^2\theta$ and $e^{c}=(f(l)-f(0))/2$ produces the desired result.
\end{proof}

\end{document}